\documentclass[11p,reqno]{amsart}

\usepackage[T1]{fontenc}
\usepackage{graphicx}
\usepackage{amssymb,amsthm,amsmath,mathrsfs,bm,braket,marginnote}
\usepackage{enumerate}
\usepackage{appendix}
\usepackage[colorlinks=true, pdfstartview=FitV, linkcolor=blue, citecolor=blue, urlcolor=blue]{hyperref}
\usepackage{multirow}

\usepackage{pgf}
\usepackage{pgfplots}
\usepackage{tikz}
\usetikzlibrary{arrows,calc}
\usepackage{verbatim}
\usetikzlibrary{decorations.pathreplacing,decorations.pathmorphing}
\usepackage[numbers,sort&compress]{natbib}
\usepackage{dsfont}

\numberwithin{equation}{section}
\newtheorem{theorem}{Theorem}[section]
\newtheorem{lemma}[theorem]{Lemma}

\theoremstyle{remark}
\newtheorem{remark}[theorem]{Remark}
\newtheorem{coro}[theorem]{Corollary}

\newtheorem{proposition}[theorem]{Proposition}

 \reversemarginpar

\begin{document}

\title[Density moments for the $d$-dimensional Fermi gas]{Moments of the ground state density for the $d$-dimensional Fermi gas in an harmonic trap} 

\author{Peter J. Forrester}
\address{School of Mathematical and Statistics, ARC Centre of Excellence for Mathematical and Statistical Frontiers, The University of Melbourne, Victoria 3010, Australia}
\email{pjforr@unimelb.edu.au}

\date{}

\dedicatory{}


\begin{abstract}
We consider properties of the ground state density for the
$d$-dimensional Fermi gas in an harmonic trap.
Previous work has shown that the $d$-dimensional Fourier transform has a very simple functional form. It is shown that this fact can be used to deduce that the density itself satisfies a third order linear differential equation, previously known in the literature but from other considerations. It is shown too how this implies a closed form expression for the $2k$-th non-negative integer moments of the density, and a second order recurrence. Both can be extended to
general Re$\, k > -d/2$. The moments, and the smoothed density, permit expansions in $1/\tilde{M}^2$, where $\tilde{M} = M + (d+1)/2$,
with $M$ denoting the shell label. The moment expansion substituted in the second order recurrence gives a generalisation of the Harer--Zagier recurrence, satisfied by the coefficients of the $1/N^2$ expansion of the moments of the spectral density for the Gaussian unitary
ensemble in random matrix theory.
\end{abstract}

\maketitle
\section{Introduction}
Consider a statistical system of $N_0$ particles confined to a 
$d$-dimensional region $V \subseteq \mathbb R^d$. Let $\Omega  
\subseteq V$ be a subdomain of $V$ with a non-zero volume in
$\mathbb R^d$ and let $N(\Omega)$ denote the expected number of 
particles in $\Omega$. The one-particle density $\rho^{N_0}(\mathbf r)$ is defined so that
$N(\Omega) = \int_{\Omega} \rho^{N_0,d}(\mathbf r) \, {\rm d}^d \mathbf r$. Our interest is in properties of the density for the
particular statistical system specified by
the ground state of $N_0$ spinless
free fermions in $\mathbb R^d$, confined (to leading order, and
upon a rescaling) to a ball about the origin by each being
subject to an isotropic harmonic potential.

Being a free system, the $N_0$-body Hamiltonian $\mathcal H_{N_0}$ for this Fermi gas is a sum
of independent one-body Hamiltonians
$$
\mathcal H_{N_0} = \sum_{j=1}^{N_0} \mathcal H^{(d)}(\mathbf x_j), \qquad
\mathcal H^{(d)}(\mathbf x) = - {1 \over 2} \Big ( \nabla^2 - || \mathbf x ||^2 \Big ),
$$
where dimensionless units are assumed. 
Writing $\mathbf x_j = (x_j^{(1)},\dots, x_j^{(d)})$,
we see that $ \mathcal H^{(d)}$ itself separates into $d$ one-dimensional operators
$$
 \mathcal H^{(d)}(\mathbf x) = \sum_{k=1}^d \mathcal H(x^{(k)}), \qquad
 \mathcal H(x) = - {1 \over 2} \Big ( {\partial^2 \over \partial x^2} - x^2 \Big ).
$$
Here $\mathcal H(x)$ is the one-dimensional harmonic
oscillator Hamiltonian, and so has normalised eigenfunctions
\begin{equation}\label{2.0}
\psi_l(x) = {1 \over \sqrt{2^l l! \pi^{1/2}}} H_l(x) e^{- x^2/2}, \qquad (l=0,1,2,\dots),
\end{equation}
where $H_l(x)$ denotes the Hermite polynomial of degree $l$,
each with corresponding eigenvalue $\varepsilon_l = l + 1/2$. 
Hence $\mathcal H^{(d)}(\mathbf x)$ 
has eigenfunctions of the explicit factorised form
\begin{equation}\label{2.1}
\psi_{\mathbf l}(\mathbf x) = \prod_{k=1}^d \psi_{l_k}(x^{(k)}), \quad
\mathbf l = (l_1,\dots, l_d),
\end{equation}
and corresponding eigenvalue $\varepsilon_{\mathbf l} = \sum_{k=1}^d
\varepsilon_{l_k}$. A many body state --- that is an
eigenfunction of $\mathcal H_{N_0}$ --- is formed as the product
of $N_0$ eigenfunctions of the form (\ref{2.1}), the $j$-th such member
depending on the co-ordinate of the $j$-th particle $\mathbf x_j$, which
must further be anti-symmetrised to give a Fermi state. 

To classify the ground state the shell label $M=0,1,2,\dots$ is introduced by the requirement
that 
\begin{equation}\label{3.0}
\sum_{j=1}^d l_j \le M .
\end{equation}
It is a straightforward exercise to show that the number of non-negative integer
arrays $\mathbf l$ satisfying this constraint is
\begin{equation}\label{3.1}
\binom{M+d}{d}.
\end{equation}
For given $M$, choosing the number of particles $N_0$ to equal (\ref{3.1}), and denoting the
arrays satisfying (\ref{3.0}) by $\mathbf l_1, \mathbf l_2,\dots, \mathbf l_{N_0}$,
one sees that the anti-symmetrisation of $\prod_{s=1}^{N_0} \psi_{\mathbf l_s}(\mathbf x_s)$
--- upon normalisation by the factor $1/\sqrt{N!}$ ---
gives the ground state eigenfunction (i.e.~eigenfunction
corresponding to the smallest eigenvalue, also referred to as
the ground state wave function), $\psi^{(0),d}(\mathbf x_1,\dots, \mathbf x_{N_0})$ say.

The one-particle density $\rho^{N_0,d}(\mathbf r)$ is computed from the ground state wave function according to
\begin{equation}\label{3.1a}
\rho^{N_0,d}(\mathbf r) = \int_{\mathbf R^d} d \mathbf x_1 \cdots  \int_{\mathbf R^d} d \mathbf x_{N_0} \,
\Big ( \sum_{l=1}^{N_0} \delta(\mathbf r - \mathbf x_l) \Big ) \Big (
\psi^{(0),d}(\mathbf x_1,\dots, \mathbf x_{N_0}) \Big )^2.
\end{equation}
The construction of $\psi^{(0),d}$ in terms of the orthonormal eigenfunctions (\ref{2.0}), together with
the  shell constraint (\ref{3.0}), shows that with
$\mathbf r = (r_1,\dots, r_d)$,
\begin{equation}\label{3.2}
\rho^{N_0,d}(\mathbf r) = \sum_{l=0}^{M} \sum_{\mathbf l : \sum_{j=1}^d l_j = l} 
\prod_{k=1}^d \Big ( \psi_{l_k}(r_k) \Big )^2.
\end{equation}
As noted in \cite{Mu04},
an early reference for this formalism in the case $d=1$ is Husimi \cite{Hu40};
our presentation has followed \cite{LG12}.
Although not immediately obvious from the form (\ref{3.2}), in keeping with the physical setting
$\rho^{N_0,d}(\mathbf r)$ is spherically symmetric and is thus a function of $|| \mathbf r || =: r$,
so we may write $\rho^{N_0,d}(\mathbf r) = \rho^{N_0,d}(r)$.

Earlier literature has uncovered a number of remarkable features of this density
\cite{LM79,BM03,ZBSB03,Mu04,SZ07,BN07,vZ12,DDMS15}. Those that impact on the present study, for which
the main theme is to investigate moments of $\rho^{N_0,d}(r)$, will be given self contained
derivations in Section \ref{S2}.
A focus on moments is suggested by the explicit form of the squared ground state eigenfunction for
$d= 1$. The above construction tells that then
$$
\psi^{(0), d=1}(x_1,\dots,x_{N_0}) = {1 \over \sqrt{N!}} {\mathcal A}{\rm symm} \, \prod_{l=1}^{N_0} \psi_{l-1}(x_l) =
{1 \over \sqrt{N!}} \det \Big [ \psi_{l-1}(x_j) \Big ]_{l,j=1}^{N_0}.
$$
Recalling (\ref{2.0}), and using the fact that for any monic polynomials $\{ p_{j-1}(x) \}_{j=1}^{N_0}$ with
$p_{j-1}$ of degree $j-1$, 
$$
\det [ p_{j-1}(x_k) ]_{j,k=1}^{N_0} = 
\det [ x_{k}^{j-1}]_{j,k=1}^{N_0} =
\prod_{1 \le j < k \le N_0} (x_k - x_j),
$$
this being the Vandermonde determinant formula,
it follows that for $d=1$ the ground state wave function squared is proportional to 
\begin{equation}\label{3.2a}
\prod_{l=1}^{N_0} e^{- x_l^2} \prod_{1 \le j < k \le N_0}(x_k - x_j)^2.
\end{equation}

The functional form (\ref{3.2a}) is  perhaps best known 
from the theory of random matrices. It occurs there
(see e.g.~\cite[Prop.~1.3.4]{Fo10}) as the eigenvalue probability
density function for $N_0 \times N_0$
Hermitian matrices $Y = {1 \over 2} (X + X^\dagger) $, where the entries of $X$ are
independent standard complex Gaussians. Equivalently, $Y$ is sampled from an ensemble of complex Hermitian matrices
with weight proportional to $e^{- {\rm Tr} \, Y^2}$, known as the Gaussian unitary ensemble (GUE).
In this setting the averages $\langle {\rm Tr} \, Y^{2k} \rangle$ correspond to the $2k$-th moment of the density.
Since the work of Br\'ezin et al.~\cite{BIPZ78} (for a textbook introduction, see e.g.~\cite[\S 1.6]{Fo10}) it has been known that 
such GUE averages have topological interpretations when expanded for large $N_0$. In fact the average
divided by $N_0$ is a polynomial of degree $k$ in $N_0$ which is even (odd) for $k$ even (odd),
\begin{equation}\label{5.1}
{1 \over N_0} \Big \langle {\rm Tr} \, Y^{2k}
\Big \rangle = \sum_{l=0}^{\lfloor k/2 \rfloor} N_0^{(k-2l)} \mu_{k,l},
\end{equation}
and moreover up to scaling the coefficients $\{ \mu_{k,l}\}$ are positive integers.

Harer and Zagier \cite{HZ86} deduced the recurrence
\begin{equation}\label{6.1}
(k+2) \mu_{k+1,l} = (k+1/2) k (k - 1/2) \mu_{k-1,l-1} + (2k+1) \mu_{k,l},
\end{equation}
subject to the initial condition $\mu_{0,0} = 1$, and boundary conditions $\mu_{k,l} = 0$ for
$k<0$ or $l< 0$ or $l > [k/2]$. For $l=0$ the recurrence simplifies to $(k+2) \mu_{k+1,0} =
(2k+1) \mu_{k,0}$ and so
\begin{equation}\label{6.2}
2^k \mu_{k,0} = {1 \over k + 1} \binom{2k}{k},
\end{equation}
which is the $k$-th Catalan number. This is in keeping with the limiting one-body density --- now
corresponding to the spectral density --- having the functional form of the Wigner semi-circle law
\begin{equation}\label{6.2a}
\lim_{N_0 \to \infty} {1 \over N_0^{1/2}} \rho^{N_0, 1} (N_0^{1/2} x) = \rho^{\rm W}(x),
\quad \rho^{\rm W}(x) := {\sqrt{2} \over \pi} (1 - x^2/2)^{1/2}
\chi_{|x| < \sqrt{2}},
\end{equation}
where $\chi_A = 1$ for $A$ true, $\chi_A = 0$ otherwise.
 In random matrix theory, there are a number of ensembles giving rise to various recursive structures of the moments
 \cite{HT03,GT05,Le04,Le09,MS11,PZ11,FL15,CMOS18,
 	CDO18,RF19,DF19}.
 
We will show in Section \ref{S3} below that the polynomial form (\ref{5.1}), upon appropriate choice of the
expansion parameter, carries over to the moments of $\rho^{N_0,d}(\mathbf r)$ in the general $d$ case. We will
show too that the coefficients satisfy a recurrence of
the same structure as (\ref{6.1}).

\begin{theorem}\label{T1}
	Let $\rho^{N_0,d}( r)$ be specified by (\ref{3.2}). Define the normalised $2k$-th radial moment by
	\begin{equation}\label{6.20}
	m_{2k}^{N_0,d} := {1 \over N_0} \int_{\mathbb R^d} || \mathbf r ||^{2k} \rho^{N_0,d}( r) \, d \mathbf r,
		\end{equation}
		where $k$ is a non-negative integer.
	Further define
	\begin{equation}\label{6.2b}
	\tilde{M} = M + {(d+1) \over 2},
	\end{equation}
	where $M$ is related to $N_0$ by the latter being equal to (\ref{3.1}). We have that in the variable
	$\tilde{M}$ these moments permit an expansion of the form (\ref{5.1}),
		\begin{equation}\label{6.2c}
		m_{2k}^{N_0,d} = \sum_{l=0}^{\lfloor k/2 \rfloor} \tilde{M}^{k-2l} \mu_{k,l}^{(d)},
	\end{equation}
	where the coefficients satisfy the generalisation of the Harer-Zagier recurrence
	\begin{equation}\label{7.1}
	(k+d+1) \mu_{k+1,l}^{(d)} = k(k+d/2)(k+d/2-1) \mu_{k-1,l-1}^{(d)} + (2k+d) \mu_{k,l}^{(d)},
	\end{equation}
	subject to the same initial and boundary conditions.
	\end{theorem}

The recurrence (\ref{7.1}) in the case $l=0$ shows that for general $d$ 
 \begin{equation}\label{6.2e}
 \mu_{k,0}^{(d)}  = 2^k {\Gamma(d/2 + k) \Gamma(d+1) \over \Gamma(d/2) \Gamma(k+d+1)} 
 \end{equation}
 (cf.~(\ref{6.2})). Making the ansatz $\mu_{k,l}^{(d)} =
 \mu_{k,0}^{(d)} p_{k,l}^{(d)}$ we see that (\ref{7.1})
 reduces to
 \begin{equation}\label{6.2x}
  p_{k+1,l}^{(d)} = {1 \over 4} k (k + d)  p_{k-1,l-1}^{(d)}
  + p_{k,l}^{(d)}
 \end{equation}
 and thus
 \begin{equation}\label{6.2y}
  p_{k,l}^{(d)} = {1 \over 4}
  \sum_{s=0}^{k-2} (s + 1) (s + 1 + d)   p_{s,l-1}^{(d)}, \qquad
  p_{k,0}^{(d)} := 1.
 \end{equation}
 In particular, we see from (\ref{6.2y}) that
 $p_{k,l}^{(d)}$ is a polynomial of degree $3l$ in $k$.
 In the case $d=1$ these polynomials
 are given explicity for $l=1,2,3$ in \cite[Theorem 7]{WF14}. For general $d$, and
 with $l=1$ we deduce from the above that 
\begin{eqnarray}\label{1.19}
\mu_{k,1}^{(d)}  = \mu_{k,0}^{(d)} {k (k-1) (2k+3d-1) \over 24}.
\end{eqnarray}

The moment formula (\ref{6.2e}) is consistent with the spectral density having as its scaled
limit the Thomas-Fermi density (see e.g.~\cite{Ca06})
\begin{equation}\label{8.d}
\lim_{N_0 \to \infty} {\tilde{M}^{1/2} \over N_0} \rho^{N_0,d}(\tilde{M}^{1/2} r) = 
\rho^{\rm TF}(r), \qquad  \rho^{\rm TF}(r) :=
  {1 \over (2 \pi)^{d/2}} {\Gamma(d+1) \over \Gamma(d/2+1)} 
  \Big ( 1 - {r^2 \over 2} \Big )^{d/2} \chi_{0 \le r < \sqrt{2}}.
\end{equation}
 Thus \cite{BM03}
\begin{equation}\label{8.e}
\int_{\mathbb R^d} || \mathbf r ||^{2k} \rho^{\rm TF}(r) \,
{\rm d}^d \mathbf r = \mu_{k,0}^{(d)}.
\end{equation}

In keeping with (\ref{6.20}) and (\ref{6.2c}) there is a generalisation of (\ref{8.e}) involving a terminating expansion in
$1/\tilde{M}^2$.

\begin{proposition}\label{P2}
	Let $\Omega_d = 2 \pi^{d/2}/\Gamma(d/2)$ denote the surface area of a $d$-dimensional
	ball of unit radius.
	The radial moments permit the $1/\tilde{M}^2$ expansion
\begin{align}\label{8.f}	
 \int_{\mathbb R^d} || \mathbf r ||^{2k}
{\tilde{M}^{2k+d/2} \over N_0} \rho^{N_0,d}(\tilde{M}^{1/2} r)
\,  {\rm d}^d \mathbf r & = \sum_{l=0}^{\lfloor k/2 \rfloor} {1 \over \tilde{M}^{2l}}  \mu_{k,l}^{(d)} \nonumber \\
& = \sum_{l=0}^{ \lfloor  k/2  \rfloor } {\Omega_d \over \tilde{M}^{2l}}
\int_{0}^{\sqrt{2}} r^{2k+d-1}
\rho^{\infty,d}_{(l)}(r) \, dr
\end{align}
(for $l \ge 1$ the measure $\rho^{\infty,d}_{(l)}(r) \, dr$ will typically contain an atom at the
upper terminal $r = \sqrt{2}$),
where $\rho^{\infty,d}_{(0)}(r) = \rho^{\rm TF}(r)$ as given in
(\ref{8.d}), and where $\{ \rho^{\infty, d}_{(l)}(r) \}$
satisfies the coupled differential equations
\begin{equation}\label{1.21}
B \rho^{\infty,d}_{(l)}(r) = A \rho^{\infty,d}_{(l-1)}(r)
\end{equation}
with 
\begin{align}\label{1.22}
A  = {1 \over 8} \Big ( {d^3 \over d r^3} +
{d - 1 \over r} {d^2 \over d r^2}- {d -1 \over r^2} {d \over dr}
\Big ), \qquad
B  = \bigg ( \Big ( {r^2 \over 2} - 1 \Big ) {d \over d r} -
{r d \over 2} \bigg )
\end{align}
and $\rho^{\infty,d}_{(-1)}(r) := 0$.
\end{proposition}
In subsection \ref{S3.4} the explicit form of $\rho^{\infty,d}_{(l)}$ for $l=1,d=1$ (an already known result)
and $l=1, d=2$ will be presented.

The result of Proposition \ref{P2} is a corollary of $\rho^{N_0,d}(r)$ itself satisfying a 3rd order differential equation.
For $d=1$ this fact was first deduced by Lawes and March \cite{LM79}. For $d=2$ the differential equation
was deduced in a paper by Minguzzi et al.~\cite{MMT01}, while for general $d$ it is due to Brack and
Murthy \cite{BM03}. In subsection \ref{s2.2} we will give a derivation of this differential equation as a
corollary of the explicit form for the Fourier
transform \cite{BN07,vZ12}
\begin{equation}\label{9.1}
\hat{\rho}^{N_0,d}(k) : = \int_{\mathbb R^d} \rho^{N_0,d}(r) e^{i \mathbf k \cdot \mathbf r} \, {\rm d}^d \mathbf r =
 e^{- k^2/4} L_M^{(d)}(k^2/2),
\end{equation}
where $k= | \mathbf k|$ and $L_n^{(\alpha)}(x)$ denotes the Laguerre polynomial. Before doing so,
in subsection \ref{s2.1},
a self contained derivation of (\ref{9.1}) will be given. It is this latter result which underpins
Theorem \ref{T1}.

\section{Characterisations of $\rho^{N_0,d}(r)$}\label{S2}
\subsection{Fourier transform: derivation of (\ref{9.1})}\label{s2.1}
From the text about (\ref{3.1}) we have that the number particles $N_0$ is related to the shell
label $M$ by 
\begin{equation}\label{NM}
N_0 = \binom{M+d}{d}.
\end{equation}
Following at first \cite{Mu04}, introduce the generating function for the densities $\rho^{N_0,d}(r)$, 
$M=0,1,\dots$ by
\begin{equation}\label{11.e}
G(r;t) := \sum_{M=0}^\infty \rho^{N_0,d}(r) t^{M}.
\end{equation}
According to (\ref{3.2}) we have
$$
G(r;t) = \sum_{M=0}^\infty \Big ( \sum_{m=0}^M \sum_{l_1 + \cdots +  l_d = m} 
| \psi_{l_1}(x_1) \cdots \psi_{l_d}(x_d) |^2 \Big ) t^M.
$$

Introduce the notation $\sigma^{(d)}_m(\mathbf r) =  \sum_{l_1 + \cdots +  l_d = m} 
| \psi_{l_1}(x_1) \cdots \psi_{l_d}(x_d) |^2$. Then we observe
\begin{align*}
(1 - t) G(r;t) & = \sum_{M=0}^\infty \Big (\sigma^{(d)}(\mathbf r) ) (t^M - t^{M+1}) \Big ) \\
& =  \sum_{M=0}^\infty \bigg ( \sum_{m=0}^M  \sigma^{(d)}_m(\mathbf r)  - 
 \sum_{m=0}^{M-1}  \sigma^{(d)}_m(\mathbf r) \bigg ) t^M \\
 & = \sum_{M=0}^\infty   \sigma^{(d)}_M(\mathbf r) t^M \\
 & = \sum_{M=0}^\infty \sum_{l_1 + \cdots +  l_d = M} 
| \psi_{l_1}(x_1) \cdots \psi_{l_d}(x_d) |^2 t^M \\
& = \prod_{j=1}^d \bigg ( \sum_{l=0}^\infty | \psi_l (x_j) |^2 t^l \bigg );
\end{align*}
this is \cite[line 1 of eq.~(7)]{Mu04}.
Making use now of the Mehler formula 
$$
e^{-(x^2 + y^2)/2} 
\sum_{n=0}^\infty {(\rho/2)^n \over n!} H_n(x) H_n(y) =
{1 \over \sqrt{1 - \rho^2}} \exp {4 xy \rho - (1 + \rho^2)(x^2 + y^2) \over 2 (1 - \rho^2)},
$$
it follows that
\begin{equation}\label{11.a}
(1 - t) G(r;t) = {1 \over \pi^{d/2}} {1 \over (1 - t^2)^{d/2}} \exp \Big ( - r^2 {1 - t \over 1 + t} \Big ),
\end{equation}
which is \cite[line 2 of eq.~(7)]{Mu04}.

Consider now the $d$-dimensional Fourier transform of $G(r;t)$,
\begin{equation}\label{12.4}
\hat{G}(k;t) := \int_{\mathbb R^d} G(r;t) e^{i \mathbf k \cdot \mathbf r} \, {\rm d}^d \mathbf r.
\end{equation}
Using the one-dimensional Fourier transform
$$
\int_{-\infty}^\infty e^{- \gamma p^2} e^{i k p} \, dp = {\pi^{1/2} \over \gamma^{1/2}} e^{- k^2/ 4 \gamma},
$$
we see from (\ref{11.a}) that
\begin{equation}\label{12.4a}
\hat{G}(k;t) = {1 \over (1 - t)^{d+1}} \exp \Big ( - {k^2 \over 4} \Big )
 \exp \Big ( - {t k^2 \over 2 (1 - t)} \Big ).
\end{equation}

Recalling now the generating function for the Laguerre polynomials 
\begin{equation}\label{12.4b}
\sum_{n=0}^\infty t^n L_n^{(\alpha)}(x) = {1 \over (1 - t)^{\alpha + 1}}\exp \Big ( - {tx \over 1 - t} \Big ),
\end{equation}
we thus have
$$
[t^{M}] \hat{G}(k;t) = \exp  \Big (- {k^2 \over 4} \Big )  L_{M}^{(d)} \Big ( {k^2 \over 2} \Big ),
$$
where in general the notation $[t^p] f(t)$ refers to the coefficient of $t^p$ in the power series expansion of $f(t)$.
According to the definition of $\hat{G}(k;t)$, with the definition of $G(r;t)$ substituted from
(\ref{11.e}), this is equivalent to the sought result (\ref{9.1}), first derived in
\cite{BN07,vZ12} using different reasoning.
\subsection{Third order linear differential equation satisfied by  $\rho^{N_0,d}(r)$}\label{s2.2}
We can use (\ref{9.1}) to first deduce that $\hat{\rho}^{N_0,d}(k)$ satisfies a particular
second order linear differential equation.

\begin{coro}\label{C1}
We have that $\hat{\rho}^{N_0,d}(k)$ satisfies
\begin{equation}\label{11.z}
{d^2 \over d k^2} f + {2 d + 1 \over k} {d \over d k} f + \bigg ( - \Big ( {k \over 2} \Big )^2 + d + 1 + 2M \bigg ) f = 0.
\end{equation}	
\end{coro}

\begin{proof}
	We know that $y = L_n^{(\alpha)}(x)$ satisfies the second
	order linear differential equation
	$$
	x y'' + (\alpha + 1 - x) y' + n y = 0.
	$$
	A simple change of variables shows that $\tilde{y} =
	 L_n^{(\alpha)}(k^2/2)$ satisfies
	 $$
	 {1 \over 2} {d^2 \over d k^2} \tilde{y} +
	 \Big ( {\alpha + 1/2 \over k} - {k \over 2} \Big )
	 {d \over d k} \tilde{y} + n \tilde{y} = 0.
	 $$
	 Introducing now $f = e^{- k^2/4} \tilde{y}$,
	 setting $n = M$, $\alpha = d$, and recalling 
	(\ref{9.1}) gives (\ref{11.z}).
	\end{proof}

When acting on a function of $k := || \mathbf k ||$, we know from the form of the Laplacian in $d$-dimensions that
\begin{equation}\label{Nk}
\nabla_{\mathbf k}^2 = {d^2 \over d k^2} + {(d-1) \over k} {d \over dk}.
\end{equation}
Hence
\begin{equation}\label{11.v1}
\bigg ({d^2 \over d k^2} + {(d-1) \over k} {d \over dk} \bigg )\hat{\rho}^{N_0,d}(k) = 
- \int_{\mathbb R^d} {\rho}^{N_0,d}(	r) r^2 e^{i \mathbf r \cdot \mathbf k} \, {\rm d}^d \mathbf r.
\end{equation}
Similarly
\begin{align}\label{11.v2}
- k^2  \int_{\mathbb R^d} {\rho}^{N_0,d}(	r) e^{i \mathbf r \cdot \mathbf k} \, {\rm d}^d \mathbf r
& = \int_{\mathbb R^d} {\rho}^{N_0,d}(	r) \nabla_{\mathbf r}^2 e^{i \mathbf r \cdot \mathbf k} \, {\rm d}^d \mathbf r \\
& = \int_{\mathbb R^d}   \Big ( \nabla_{\mathbf r}^2{\rho}^{N_0,d}(	r) \Big ) e^{i \mathbf r \cdot \mathbf k} \, {\rm d}^d \mathbf r,
\end{align}
where the second equality follows by integration by parts.
As a consequence of Corollary \ref{C1} it therefore follows that
\begin{equation}\label{11.v3}
 \int_{\mathbb R^d} \bigg (\Big ( {1 \over 4} \nabla_{\mathbf r}^2  - r^2  + d + 1 + 2M \Big ) {\rho}^{N_0,d}(	r) \bigg ) e^{i \mathbf r \cdot \mathbf k} \, {\rm d}^d \mathbf r 
 + { d + 2 \over k} {d \over d k}  \hat{\rho}^{N_0,d}(k)
 = 0.
 \end{equation}
 
 In deducing (\ref{11.v3}) from knowledge that $\hat{\rho}^{N_0,d}(k)$ satisfies (\ref{11.z}), only its definition as
 a $d$-dimensional Fourier transform has been used. 
 If, in addition,
 use is made of special properties of $\hat{\rho}^{N_0,d}(k)$,
 it is possible to also write the final term in (\ref{11.v3}) as the $d$-dimensional Fourier transform of a particular
 function of $r$.
 
 \begin{lemma}
 We have
 \begin{equation}\label{11.v4}
 - {1 \over k} {d \over d k}  \hat{\rho}^{N_0,d}(k) =  \int_{\mathbb R^d} \bigg ( \int_r^\infty s \rho^{N_0,d}(s) \, ds \bigg )
 e^{i \mathbf r \cdot \mathbf k} \, {\rm d}^d \mathbf r .
 \end{equation}
 \end{lemma}

 \begin{proof}
 It follows from (\ref{11.a}) that	
 $$
  \int_r^\infty s G(s;t) \, ds =  {1 \over 2}  {1 + t \over 1 - t} \hat{G} (r;t).
  $$
  As a consequence, with $\hat{H}(k;t)$ denoting the $d$-dimensional Fourier transform of $ \int_r^\infty s G(s,t) \, ds$ we have
 \begin{align*}
 \hat{H}(k;t) & =   {1 \over 2}  {1 + t \over 1 - t} \hat{G} (r;t) \\
 & = \Big ( {1 \over 2} + {t \over 1 - t} \Big ) {1 \over (1 - t)^{d+1}} \exp \Big ( - {t k^2 \over 2 (1 - t)} \Big ) \exp \Big ( - {k^2 \over 4}\Big ),
 	\end{align*}
 	where the second line follows from (\ref{12.4a}). Use now of the Laguerre polynomial generating function (\ref{12.4b}) shows
  \begin{align*}
[t^M] \hat{H}(k;t) & = 	{1 \over 2} \exp \Big ( - {k^2 \over 4} \Big ) L_M^{(d)}\Big ( {k^2 \over 2} \Big )+ 	 \exp \Big ( - {k^2 \over 4} \Big )
L_M^{(d+1)}\Big ( {k^2 \over 2} \Big ) \\
& = - {1 \over k} {d \over d k} \bigg ( \exp \Big ( - {k^2 \over 4} \Big ) L_M^{(d)}\Big ( {k^2 \over 2} \Big ) \bigg ) \\
& = - {1 \over k} {d \over d k} \hat{\rho}^{N_0,d}(k),
\end{align*}
where the final equality makes use of (\ref{9.1}). Thus (\ref{11.v4}) is established.
	\end{proof}

Substituting (\ref{11.v4}) in (\ref{11.v3}) gives an integro-differential equation for ${\rho}^{N_0,d}(r)$, which
itself is equivalent to a third order linear equation, first derived for general $d$ in \cite{BM03} using different reasoning.

\begin{proposition}
	We have
	\begin{equation}\label{11.d3}
	- {1 \over 8} \nabla_{\mathbf r}^2 {\rho}^{N_0,d}(r) + {1 \over 2} r^2  {\rho}^{N_0,d}(r) + {d + 2 \over 2}
	\int_r^\infty s  {\rho}^{N_0,d}(s) \, ds = \Big ( M + (d+1/2) \Big ) {\rho}^{N_0,d}(r),
	\end{equation}
	which upon use of (\ref{Nk}) with $\mathbf k$ replaced by $\mathbf r$, and a further differentiation with respect to $r$, implies
		\begin{equation}\label{11.d4}
		\bigg ( - {1 \over 8} {d^3 \over d r^3} - {1 \over 8} {d - 1 \over r} {d^2 \over d r^2} +
		\Big ( {d - 1 \over 8 r^2} + {r^2 \over 2} - (M + (d+1)/2) \Big ) {d \over d r} -
		{ d  \over 2}r \bigg )  {\rho}^{N_0,d}(r) = 0.
		\end{equation}
		\end{proposition}
	
	Use will be make of (\ref{11.d4}) to derive the
	recursive differential relations (\ref{1.21}) in
	subsection \ref{S3.3}. In the Appendix we will show
	how a scaling of this equation near the boundary
	of the leading order support (referred to as soft edge scaling) can be used to deduce that the soft edge density
	satisfies a particular third order differential equation.
	
	\section{Properties of the moments}\label{S3}
	\subsection{Closed form and recurrence}
	The normalised $2k$-th radial moment $m_{2k}^{N_0, d}$ is specified by (\ref{6.20}), where it is further required
	that $k$ be a non-negative integer. The significance of this
	latter requirement follows from the fact that, as a 
	consequence of (\ref{9.1}), $m_{2k}^{N_0, d}$ then admits a closed form evaluation.
	
	\begin{proposition}
		In the above setting
		\begin{equation}\label{mk}
		m_{2k}^{N_0, d} = {1 \over N_0}{\Gamma(d/2 + k) \over \Gamma(d/2)}
		\sum_{l=0}^k \binom{k}{k-l} \binom{M+d}{d+l}2^l.
		\end{equation}
		Furthermore, $\{ 	m_{2k}^{N_0, d} \}$ satisfy the
		second order recurrence
			\begin{equation}\label{mk1}
			{2(k+d+1) \over (2k+d)} m_{2k+2}^{N_0, d} =
			(2M + d + 1) m_{2k}^{N_0, d} + k (k + d/2 -1)m_{2k-2}^{N_0, d},
			\end{equation}
			valid for $k=0,1,\dots$ and subject to the initial
			condition $m_0^{N_0, d} = 1$.
		\end{proposition}
	
	\begin{proof}
		First, to avoid confusion between the use of $k$ in
		(\ref{9.1}), and its use in
		$m_{2k}^{N_0, d}$ as defined by (\ref{6.20}), we will
		replace the latter by $p$ in the subsequent working,
		and so consider 	$m_{2p}^{N_0, d}$ for
		$p \in \mathbb Z_{\ge 0}$. Comparing the definition of
		the latter with
		the definition of $\hat{\rho}^{N_0,d}(k)$ in (\ref{9.1}), we see by setting $\mathbf k = (k,0,\dots,0)$,
		power series expanding the exponential, and changing to polar coordinates, that 
			\begin{equation}\label{fp0}
		m_{2p}^{N_0, d} = {1 \over N_0} f_{p,d} (2p)! (-1)^p [k^{2p}]
		\hat{\rho}^{N_0,d}(k)
		\end{equation}
		where, with $d \Omega_d$ denoting the infinitesimal volume
		element for the angular contribution to the Lebesgue measure on $\mathbb R^d$,
		\begin{equation}\label{fp}
		f_{p,d} = {\int d \Omega_d \over \int \cos^{2p} \theta_1 \,
		d \Omega_d}
	\end{equation}
	
	To evaluate $f_{p,d}$, we first note that
	$$
	e^{- r^2/4} = {\int e^{-k^2} e^{-i \mathbf k \cdot \mathbf r} \,
	d \mathbf k \over \int e^{-k^2} \, d \mathbf k} =
\sum_{p=0}^\infty {(-1)^p r^{2p} \over (2p)!}
{\Gamma (d/2 + p) \over \Gamma(d/2)} {1 \over f_{p,d}},
$$
where the first equality can be verified by evaluating the
integrals, while the second follows by first performing
a power series expansion of the complex exponential, then
evaluating the integral in the numerator and denominator
by changing to polar coordinates.
Now power series expanding the LHS and equating coefficients
of $r^{2p}$ gives
\begin{equation}\label{fp1}
{1 \over f_{p,d}} = {\Gamma(d/2) \over \Gamma(d/2 + p)}
{(2p)! \over 4^p p!}.
\end{equation}

According to the RHS of (\ref{9.1}), upon recalling the series
form of the Laguerre polynomial
$$
L_n^\alpha (x) = \sum_{k=0}^n (-1)^k {1 \over k!}
\binom{n + \alpha}{n-k} x^k,
$$
we have from the formula for the coefficients in the product
of two power series
\begin{equation}\label{fp2}
{1 \over N_0}[k^{2p}]
\hat{\rho}^{N_0,d}(k) = {1 \over N_0} \sum_{l=0}^p (-1)^l {1 \over l!}
\binom{M+d}{M-l} {1 \over 2^l} (-1)^{p-l} {1 \over 4^{p-l}}
{1 \over (p-l)!},
\end{equation}
where use too has been made of (\ref{NM}). Substituting
(\ref{fp1}) and (\ref{fp2}) in (\ref{fp0}) we obtain
(\ref{mk}).

In relation to the recurrence (\ref{mk1}), we begin by noting
from (\ref{fp0}) that
\begin{equation}\label{fp3}
{1 \over N_0}
\hat{\rho}^{N_0,d}(k) = \sum_{p=0}^\infty { (-1)^p m_{2p}^{N_0, d} k^{2p} \over  f_{p,d} (2p)!}.   
\end{equation}
Substituting in the differential equation (\ref{11.z}) and
manipulating the sums gives
\begin{multline*}
- \sum_{p=0}^\infty {(-1)^p k^{2p} m_{2p+2}^{N_0, d} \over (2p)!
f_{p+1,d}} - (2d+1) \sum_{p=0}^\infty {(-1)^p k^{2p} 
m_{2p+2}^{N_0, d} \over (2p+1) (2p)! f_{p+1,d}} \\
+ {1 \over 4}  \sum_{p=0}^\infty {(-1)^p k^{2p} m_{2p-2}^{N_0, d} \over (2p-2)!
	f_{p-1,d}} + (d+1 + 2M) \sum_{p=0}^\infty {(-1)^p k^{2p}
m_{2p}^{N_0, d} \over (2p)! f_{p,d} } = 0.
\end{multline*}
Equating coefficients of $(-1)^p k^{2p}/(2p)!$ then shows
$$
(2p+2d+2) { m_{2p+2}^{N_0, d} \over f_{p+1,d}} =
(2p+1)p(p-1/2)  { m_{2p-2}^{N_0, d} \over f_{p-1,d}} +
(2p+1) (d+1+2M)  { m_{2p}^{N_0, d} \over f_{p,d}}.
$$
Substituting for $\{f_{p,d} \}$ as specified by (\ref{fp1}),
 (\ref{mk1}) follows.
		\end{proof}
	
	Setting $k=1$ and $k=2$ in (\ref{mk}) gives, upon recalling the value of $N_0$ from (\ref{NM}), and simplifying
	\begin{align}
	 m_{2}^{N_0, d} & =  {d(M + (d+1)/2) \over d + 1} \\
	 m_4^{N_0, d} & = {d \over (d+1) } \Big ( (M + (d+1)/2)^2 + (d+1)/4
	 \Big ).
	 \end{align}
	 Note that the dependence on $N_0$ is through the quantity
	 $\tilde{M} := M + (d+1)/2$, as is consistent
	with  (\ref{6.2c}) and (\ref{7.1}) for $k=1$ and 2.
	 In fact both can now be established in the general $k$ case.
	 
	 \medskip
	 \noindent
	 {\it Proof of Theorem \ref{T1}.} \quad That the dependence
	 on $N_0$ in $m_{2k}^{N_0, d}$ is through the quantity
	 $\tilde{M} := M + (d+1)/2$ for general $k$ follows from the recurrence (\ref{mk1}), which moreover implies that $m_{2k}^{N_0, d}$ is even (odd) in $\tilde{M}$ for $k$ even (odd) and so permits the expansion (\ref{6.2c}).
	 Substituting (\ref{6.2c}) in (\ref{mk1}) implies
	  (\ref{7.1}). \hfill $\square$
	  
	  \begin{remark}
	  	The explicit formula (\ref{1.19}) shows that $\mu_{k,1}^{(d)}$ vanishes for $k=0, 1$. More generally
	  	$\mu_{k,l}^{(d)}$ vanishes for $k=0,\dots,l$, as is consistent with 
	  	(\ref{6.2c}).
	  	\end{remark}
	
\subsection{The moments for general $k$}
It is known for $d=1$ that the moments (\ref{mk}) can be written
as a Gauss hypergeometric function \cite{WF14,CMOS18},
\begin{equation}\label{fp4}
2^k m_{2k}^{N_0,d=1} = {(2k)! \over 2^k k!} \,
{}_2 F_1 ( -k, - N_0 + 1; 2; 2).
\end{equation}
As emphasised in \cite{CMOS18} the hypergeometric function is a
polynomial in $k$ of degree $N_0 - 1$, and as such has a unique
analytic continuation from the integers to general 
$k \in \mathbb C$. In particular, (\ref{fp4}) therefore evaluates
the integral (\ref{6.20}) in the case $d=1$ for all $k \in
\mathbf C$, ${\rm Re} \, k > -1/2$, where this latter condition is
required for convergence (note that the factor $(2k)!$ in
(\ref{fp4})  diverges as $k \to  -  1/2$ from above).

The general $d$ case also admits an evaluation well defined in the complex plane.

\begin{proposition}
	We have
	\begin{equation}\label{fp5}
	m_{2k}^{N_0,d} = { \Gamma(d/2 + k) \over  \Gamma(d/2) }\,
	{}_2 F_1 ( -k, - M; d+1; 2),
	\end{equation}
	which is an analytic function of $k$ for ${\rm Re} \, k
	> - {d \over 2}$.
	\end{proposition}

\begin{proof}
	Straightforward manipulation of  (\ref{mk}) gives
	$$
		m_{2k}^{N_0,d} = { \Gamma(d/2 + k) \over  \Gamma(d/2) }\,
	\sum_{l=0}^k {(-k)_l (-M)_l \over l! (d + 1)_l} 2^l,
	$$
	where $(a)_n := a (a+1) \cdots (a+n-1)$ denotes the rising
	factorial Pochhammer symbol. The sum herein is precisely the series form of the Gauss hypergeometric function in (\ref{fp5}). This is a polynomial in $k$ of degree $M$, and
	so is well defined for general complex $k$. The domain of analyticity is thus fully determined by the factor
	$ \Gamma(d/2 + k)$, for which the singularity with largest
	real part occurs at $k=-d/2$.
	\end{proof}

It has been observed in \cite{CMOS18} that the $d=1$ formula
(\ref{fp4}) allows the corresponding general complex moment 
to be expressed in terms of a particular discrete orthogonal polynomial ---
the Meixner polynomial ${\rm M}_n(x;\alpha,c)$ --- in two different
ways. This carries over to the general $d$ case.

\begin{proposition}
	We have
\begin{equation}\label{fp6}
m_{2k}^{N_0,d} =
{ \Gamma(d/2 + k) \over  \Gamma(d/2) }\,
{\rm M}_k(M;d+1,-1) = { \Gamma(d/2 + k) \over  \Gamma(d/2) }\,
{\rm M}_M(k;d+1,-1).
\end{equation}
\end{proposition}

\begin{proof}
	These forms follow from \eqref{fp5}, the relationship between
	the Gauss ${}_2 F_1$ function and the Meixner polynomial
	\begin{equation}\label{fp7}
	{\rm M}_n(x;\alpha,c) = {}_2 F_1 \Big (-n,-x,\alpha; 1 - {1 \over c} \Big ),
	\end{equation}
	and the symmetry ${}_2 F_1(a,b,c;x) = {}_2 F_1(b,a,c;x)$.
	\end{proof}

\begin{coro}
	With $N_0$ determined in terms of $M$ 
	by (\ref{NM}), write $N_0 = N_0(M)$.
	The moments exhibit the reflection formula
		\begin{equation}\label{fp8}
		{1 \over \Gamma (d/2 + k)}
		m_{2k}^{N_0(M),d} = {1 \over \Gamma (d/2 + M)}
		m_{2M}^{N_0(k),d}.
		\end{equation}
		\end{coro}
	
	\begin{proof}
		This is immediate from (\ref{fp6}). 
\end{proof}

\begin{remark} In the case $d=1$, (\ref{fp8}) can
	be found in \cite[Corollary 4.3]{CMOS18}.
	\end{remark}

The Meixner polynomials satisfy the second order difference
equation
\begin{equation}\label{fp9}
c (x + \beta) M_n(x+1;\beta,c) = \Big ( n(c-1) + x + (x + \beta)c \Big ) M_n(x;\beta,c) -
x M_n(x-1;\beta,c).
\end{equation}
Substituting according to the second equality of (\ref{fp6}) we reclaim the moment recurrence
(\ref{mk}). However, in (\ref{fp9}) $k$ is not restricted to be a non-negative integer, so
its validity is now extended to all ${\rm Re} \, k
> - {d \over 2}$.

\begin{remark}\label{R3.7}
	It is also possible to deduce that (\ref{mk}) is valid
	for all ${\rm Re} \, k
	> - {d \over 2}$ by
	using the third order differential equation (\ref{11.d4}).
	Thus we multiply this equation on the left by $r^{d+2k}$, and integrate
	over $r$ from 0 to $\infty$, making use of integration by parts.
	A second order recurrence in $k$ for 
	$$
	r_{2k}^{N_0,d} :=
	\int_0^\infty r^{2k+d-1} \rho^{N_0,d}(r) \, dr
	$$
	results. But the use of polar coordinates gives
	$
	m_{2k}^{N_0,d} = \Omega_d  r_{2k}^{N_0,d} / N_0$, 
	where
	$\Omega_d$ is the surface area of a $d$-dimensional unit ball,
	so this is equivalently a second order recurrence in $k$ for
	$\{m_{2k}^{N_0,d}\}$, which in fact is precisely (\ref{mk}).
	\end{remark}

\subsection{Expansion of the density: proof of 
	Proposition \ref{P2}}\label{S3.3}
As observed in Remark \ref{R3.7} there is an equivalence between the third order differential equation for $\rho^{N_0,d}$ and the second order recurrence for $\{m_{2k}^{N_0,d}\}$. Due to this,
the expansion (\ref{6.2c}) for the moments can equivalenty be studied via an expansion of the density
\begin{equation}\label{fp10}
{\tilde{M}^{d/2} \over N_0} \rho^{N_0,d}(\tilde{M}^{1/2} r) \doteq 
\sum_{l=0}^\infty \tilde{M}^{-2l}\rho_{(l)}^{\infty,d}(r).
\end{equation}
Here the use of $\doteq$ is to indicate that both sides must
be smoothed by integrating over a suitable test function.
In particular,  with $\{ \mu_{k,l}^{(d)} \}$ as specified
by (\ref{6.2c}), we must have
\begin{equation}\label{3.16a}
\Omega_d \int_{0}^{\sqrt{2}} r^{2k+d-1}
 \rho_{(l)}^{\infty,d}(r) \, dr = \mu_{k,l}^{(d)}.
\end{equation}
 This is the second equality in (\ref{8.f}).
 
 In relation to the recurrence (\ref{1.21}), with the notation of
 (\ref{1.22}) we see that upon the change of variables
 $ r \mapsto \tilde{M}^{1/2} r$, the differential equation
 (\ref{11.d4}) can be written
 \begin{equation}\label{3.17a}
 B \rho^{N_0,d} (\tilde{M} r) = {1 \over \tilde{M}^2}
 A \rho^{N_0,d} (\tilde{M} r).
 \end{equation}
 The recurrence  follows from this
 by substituting (\ref{fp10}).
 
 \medskip
 In the case $l=0$ the RHS of (\ref{1.21}) vanishes,
 and after minor manipulation, we obtain the
 differential equation for $\rho^{\infty,d}_{(0)}$,
 $$
 {d \over dr} \log \rho^{\infty,d}_{(0)}(r)
 = \Big ( {r d \over 2} \Big )
 {1 \over r^2/2 - 1}.
 $$
 This equation has the unique non-negative continuous
 solution
 $$
 \rho^{\infty,d}_{(0)}(r) = c_0 (1 - r^2/2)^{d/2}
 \chi_{|r| < \sqrt{2}}.
 $$
 Fixing $c_0 = c_0(d)$ by the requirement that
 $\int_{\mathbb R^d}  \rho^{\infty,d}_{(0)}(r)
 \, d^d \mathbf r = 1$ reclaims
 $\rho^{\rm TF}(r)$ as specified by (\ref{8.e}).
 
 \subsection{Computation of $\rho_{(l)}^{\infty,d}$ for $(l,d) = (1,1)$ and $(1,2)$}\label{S3.4}
 Beyond the case $l=0$, we expect $\rho_{(l)}^{\infty,d}(r)$
 to have singularities at the boundary $r=\sqrt{2}$ of the support, which may involve delta functions.
 
 One way to probe such singularities is to transform
 (\ref{1.21}) from a differential equation for the
 densities to a differential equation for the corresponding
 Stieltjes transform,
 \begin{equation}\label{WS}
 W_{(l)}^{\infty, d}(z) := {\Omega_d \over 2}
 \int_{-\infty}^\infty {
 \rho_{(l)}^{\infty,d}(|r|) \over z - r} |r|^{d-1} \, dr
= {1 \over z} \sum_{k=0}^\infty {\mu_{k,l}^{(d)}
	\over z^{2k}},
\end{equation}
where the final equality follows by an application of the
geometric series formula and use of (\ref{3.16a});
see \cite{GT05}, \cite{HT12} and \cite{WF14} for the case $d=1$.
In \cite{HT12}, the explicit functional form
(albeit with some coefficients specified recursively) of
$\{ W_{(l)}^{\infty,1}(z)\}$ was presented, and we
read off in particular that
\begin{equation}\label{Wz}
 W_{(0)}^{\infty,1}(z) = \Big ( z - \sqrt{z^2 - 2} \Big ),
 \qquad W_{(1)}^{\infty,1}(z) = {1 \over 4}
 (z^2 - 2)^{-5/2}.
 \end{equation}
 The Sokhotski-Plemelj formula can be used to invert (\ref{Wz}), reclaiming $\rho^{\rm TF}(r)|_{d=1}$
 as specified in (\ref{8.d}) for $ \rho_{(0)}^{\infty,1}(r)$,
 and giving
 \begin{equation}
 \rho_{(1)}^{\infty,1}(r) = {1 \over 4 \pi }
 {1 \over (2 - r^2)^{5/2}} \chi_{0 < r < \sqrt{2}}.
 \end{equation}
Although this as a non-integrable singularity as 
$r \to \sqrt{2}^-$, as noted in \cite{WF14} it can be
integrated against power functions using the Euler
beta integral. Doing this,  we see that the LHS of
(\ref{3.16a}) in the case $d=1$, $l=1$
agrees with the RHS as specified by (\ref{6.2e})
and (\ref{1.19}).

In fact the simple structure of (\ref{1.19}) suggests
an alternative approach to the computation of 
$W_{(1)}^{\infty,1}(z)$, which we will carry out in the
case $d=2$. Although this approach is generally applicable,
there are simplifying features for $d=2$ which aid in the calculation. For guidance,
we begin by manipulating (\ref{WS}) so that the integration is over the positive half line, and then substitute
$s=r^2$ to obtain
\begin{equation}\label{W2z}
W_{(l)}^{\infty,2}(z) = z \int_0^\infty {\tilde{\rho}_{(l)}^{\infty,2}(s) \over z^2 - s} \, ds,
\qquad \tilde{\rho}_{(l)}^{\infty,2}(s) :=
{\Omega_2 \over 2}{\rho}_{(l)}^{\infty,2}(r) \Big |_{r^2 = s}.
\end{equation}
In relation to $\tilde{\rho}_{(l)}^{\infty,2}(s)$,
with $l=0$ use of (\ref{8.d}) shows that
\begin{equation}\label{W2c}
\tilde{\rho}_{(0)}^{\infty,2}(s) =
(1 - s/2) \chi_{0 < s < 2}.
\end{equation}
The formula (\ref{W2z}) suggests working with the modified
Stieltjes transform
\begin{equation}\label{GG}
G_{(l)}(x) := {1 \over z} W_{(l)}^{\infty, 2}(z) \Big |_{z^2 = x} = \int_0^\infty {{\rho}_{(l)}^{\infty,2}(s) \over
	x - s} \, ds = {1 \over x} \sum_{k=0}^\infty
{\mu_{k,l}^{(2)} \over x^k}.
\end{equation}
We see from (\ref{W2c}) that for $d=2$ $\tilde{\rho}_{(0)}^{\infty,d}(s)$ has a very simple form.
This, together with the elimination of the
term $|r|^{d-1}$ seen in (\ref{WS}) upon the
change of variable, allows for an efficient computation of
$G_{(1)}(x)$, and consequently of ${\rho}_{(1)}^{\infty,2}(x)$.

\begin{proposition}
	With $G_{(l)}(x)$ specified by (\ref{GG}) we have
	\begin{align}
	G_0(x) & = - \log \Big ( 1 - {2 \over x} \Big )
	+ {x \over 2} \log \Big ( 1 - {2 \over x} \Big ) + 1, 
\label{G1a} \\
	G_1(x) & = - {1 \over 12} \bigg ( 3  \log \Big ( 1 - {2 \over x} \Big )
	- {3 x \over 2} \log \Big ( 1 - {2 \over x} \Big ) 
	+ {3 \over x - 2} - {4 \over (x-2)^2} -3 \bigg ). \label{G1b}
	\end{align}
	Hence
	\begin{equation}\label{G1c}
	\tilde{\rho}_{(1)}^{\infty,2}(x) = {1 \over 4} \Big ( 1 - {x \over 2}
	\Big )  \chi_{0 < x < 2} - {1 \over 4} \delta(x - 2) - {1 \over 3} \delta'(x-2).
	\end{equation}
	\end{proposition}

\begin{proof}
	According to (\ref{GG}) and (\ref{6.2e})
	\begin{equation}\label{G2}
	G_{0}(x)  = {1 \over x} \sum_{k=0}^\infty
	{\mu_{k,0}^{(2)} \over x^k} =
	\sum_{k=0}^\infty {2^{k+1} \Gamma(1+k) \over \Gamma(3+k)}
	x^{-k-1}.
	\end{equation}
Simple manipulation of the series shows that it
agrees with the analogous series implied by
(\ref{G1a}).

From the first equality in (\ref{G2}) we deduce
\begin{align}
x^{5/2} {d \over dx} x^{-5/2} {d \over dx} x^2
{d \over dx} x G_{(0)}(x) &= \sum_{k=0}^\infty 
 \mu_{k,0}^{(2)} (-k)(-k+1)(-k-5/2) x^{-k-1} 
 \nonumber \\
 & =
 -12 \sum_{k=0}^\infty 
 \mu_{k,1}^{(2)} x^{-(k+1)} = -12 G_{(1)}(x),
 \end{align}
		where the second equality follows from (\ref{1.19}).
		Substituting (\ref{G1a}) in the LHS and computing
		the derivatives gives (\ref{G1b}).
		
		Recall now the second equality in (\ref{GG}).
		Extracting $\tilde{\rho}_{(1)}^{\infty,2}(x)$ by
		application of the  Sokhotski-Plemelj formula in relation to the first two terms, and an inspection
		of the remaining terms, we deduce (\ref{G1c}).
		\end{proof}
	
	\begin{remark}
		We see that for $0 \le x < 2$, the functional forms
		of both $\tilde{\rho}_{(0)}^{\infty,2}(x)$ and
		$\tilde{\rho}_{(1)}^{\infty,2}(x)$ are proportional
		to $(1-x/2)$. This can be understood from the fact that this functional form, with $x$ replaced by $r^2$, is annihilated by both the differential operators $A$ and $B$ in (\ref{1.21}) with $d=2$, telling us furthermore
		that the same is true of
		$\tilde{\rho}_{(l)}^{\infty,2}(x)$ for each $l=0,1,2,\dots$. It is also true of $\rho^{N_0,2}(r)$,
		meaning that $(1 - r^2/2N_0)$ is an exact solution of
		the third order equation (\ref{11.d4}). However for finite $N$ and $0 \le r < 2 \sqrt{N}_0$, the density is not proportional to this one solution
		even though this is the limiting Thomas-Fermi form.
		Rather the asymptotic analysis of \cite{Mu04} shows that there are oscillatory terms, as well as edge effects; in relation to the latter, see as Appendix
		\ref{A}.
		\end{remark}
	
	\subsection*{Acknowledgements}
	This research is part of the program of study supported
	by the Australian Research Council Centre of Excellence ACEMS, and the project DP170102028. The support by D.~Wang (NUS), G.~Akemann (Bielefeld) and D.~Dai (City University) for visits to their home institutions over the time span of the project is much appreciated, as is correspondence from K.~Bencheikh, and the careful reading of a referee.
	
	\appendix
	\section{Soft edge scaling} \label{A}
	The scaling in (\ref{fp10}) relates to what is termed
	the global density, whereby to leading order the support is a ball of finite radius (here $\sqrt{2}$).
	Beginning with the work \cite{Mu04}, and extended in
	\cite{DDMS15,DDMS16}, there is interest in the functional form of the density in the neighbourhood of the boundary of the support, and with a scale so that the spacing between particles in this region is of order unity.
	This can be achieved by changing variables
	\begin{equation}\label{A1}
		r \mapsto  \sqrt{2 \tilde{M}} \Big ( 1 + {s \over 2 \tilde{M}^{2/3}} \Big )
		\end{equation}
		and introducing
		\begin{equation}\label{A2}
		\rho^{\rm edge}(s) :=  \lim_{\tilde{M} \to \infty}{1 \over \sqrt{2} \tilde{M}^{1/6}} \rho \Big (\sqrt{2 \tilde{M}} \Big ( 1 + {s \over 2 \tilde{M}^{2/3}} \Big ).
		\Big )
		\end{equation}
		Making this change of variable in the differential equation (\ref{11.d4}) and taking the limit 
		$\tilde{M} \to \infty$  gives
			\begin{equation}\label{A3}
			- {1 \over 4} {d^3 \over d s^3}
			{\rho}^{\rm edge}(s) +
			\Big ( s {d \over ds} - {d \over 2} \Big )
			{\rho}^{\rm edge}(s) = 0.
			\end{equation}

		In fact this is precisely \cite[Eq.~(207)]{DDMS16},
		where it is observed to be satisfied by
		\begin{equation}\label{A4}
		F_d(s) = {1 \over \Gamma(d/2+1) 2^{4/3} \pi^{d/2}}
		\int_0^\infty u^{d/2} {\rm Ai}(u + 2^{2/3}s) \, du,
		\end{equation}
		where here ${\rm Ai}(x)$ refers to the Airy function, derived in \cite{DDMS15} from an asymptotic analysis of 
		(\ref{11.e}).

\end{document}